\documentclass[12pt]{amsart}

\usepackage[hmargin=1.2in,vmargin=1.2in]{geometry}
\usepackage{amsmath,amsthm,amssymb,amsfonts,verbatim}
\usepackage{graphicx}
\usepackage[colorlinks, linkcolor=blue,  anchorcolor=blue,
            citecolor=blue
            ]{hyperref}
\usepackage[hmargin=1.2in,vmargin=1.2in]{geometr y}
\usepackage{booktabs}

\title[Binary sequences with a low correlation]{Binary sequences with a low correlation via cyclotomic function fields with odd characteristics}
\author{Lingfei Jin}\address{Shanghai Key Laboratory of Intelligent Information Processing, School of Computer Science, Fudan University, Shanghai 200433, China.} \email{lfjin@fudan.edu.cn}
\author{Liming Ma}\address{School of Mathematical Sciences, University of Science and Technology of China, Hefei 230026, China}\email{lmma20@ustc.edu.cn}
\author{Chaoping Xing} \address{School of Electronic Information and Electric Engineering, Shanghai Jiao Tong University, China 200240}\email{xingcp@sjtu.edu.cn}

\newtheorem{lemma}{Lemma}[section]
\newtheorem{theorem}[lemma]{Theorem}

\newtheorem{prop}[lemma]{Proposition}
\newtheorem{ex}[lemma]{Example}

\theoremstyle{remark}

\renewcommand{\epsilon}{\varepsilon}
\renewcommand{\le}{\leqslant}
\renewcommand{\ge}{\geqslant}

\def\PGL{{\rm PGL}}
\def\Gal{{\rm Gal}}

\newcommand{\vnote}[1]{}

\def\PP{\mathbb{P}}

\def\F{\mathbb{F}}
\def\Z{\mathbb{Z}}

\def \mL {\mathcal{L}}

\def \mS {\mathcal{S}}

\def \Xi {{X^{[i]}}}

\newcommand{\Ga}{\alpha}
\newcommand{\Gb}{\beta}

\newcommand{\Gl}{\lambda}    
     
\newcommand{\Gs}{\sigma}

\newcommand{\s}{\sigma}

\def\PGL{{\rm PGL}}

\def \bs {{\bf s}}

\def \bu {{\bf u}}
\def \bv {{\bf v}}

\def\mS{{\mathcal S}}
\def\Aut {{\rm Aut }}

\def\Gal{{\rm Gal}}

\onecolumn

\begin{document}
\maketitle

\begin{abstract}
Sequences with a low correlation have very important applications in communications, cryptography, and compressed sensing.
In the literature, many efforts have been made to construct good sequences with various lengths where binary sequences attracts great attention. As a result, various constructions of good binary sequences have been proposed.
However, most of the known constructions made use of the multiplicative cyclic group structure of finite field $\F_{p^n}$ for a prime $p$ and a positive integer $n$. In fact, all $p^n+1$ rational places including the place at infinity of the rational function field over $\F_{p^n}$ form a cyclic structure under an automorphism of order $p^n+1$. In this paper, we make use of this cyclic structure to provide an explicit construction of binary sequences with a low correlation of length $p^n+1$ via cyclotomic function fields over $\F_{p^n}$ for any odd prime $p$. Each family of binary sequences has size $p^n-2$ and its correlation is upper bounded by $4+\lfloor 2\cdot p^{n/2}\rfloor$. To the best of our knowledge, this is the first construction of binary sequences with a low correlation of length $p^n+1$ for odd prime $p$.
Moreover, our sequences can be constructed explicitly and have competitive parameters.
\end{abstract}

\section{Introduction}
Sequences with a low correlation have important applications in code-division multiple access (CDMA), spread spectrum systems, and broadband satellite communications \cite{AP07, CHPW15, K10, LHL,Lin,16ZXE}.
For these applications, sequences with good parameters, such as a low correlation (including autocorrelation and cross-correlation) and a large family size, are highly desired. In the literature, many families of binary sequences with good properties have been constructed from various methods. The Gold sequences is the best known binary sequences with optimal correlations  which can be obtained by EXOR-ing two $m$-sequences \cite{m} of the same length with each other \cite{Span}.
Kasami sequences can be constructed using m-sequences and their decimations \cite{Kas}.
Besides, there are some known families of sequences of length  $2^n-1$ with good correlation properties, such as  bent function sequences \cite{Bent}, No sequences \cite{No}, Trace Norm sequences \cite{Klap}.
In 2011, Zhou and Tang generalized the modified Gold sequences and obtained binary sequences with length $2^n-1$ for $n=2m+1$, size $2^{\ell n}+\cdots+2^n-1$, and correlation $2^{m+\ell}+1$ for each $1\le \ell \le m$ \cite{ZT}.
In addition to the above mentioned sequences of length $2^n-1$, there are also some work devoted to the constructions of good binary sequences with length of forms such as $2(2^n-1)$ and $(2^n-1)^2$.
In \cite{Uda}, Udaya and Siddiqi obtained a family of $2^{n-1}$ binary sequences, length $2(2^n-1)$  satisfying the Welch bound for odd $n$.
This family was extended to a larger one with  $2^n$ sequences and the same correlation in \cite{Tang3}. In \cite{Tang2}, the authors presented two optimal families of sequences of length $2^n-1$ and $2(2^n-1)$ for odd integer $n$.
In 2002, Gong presented a construction of binary sequences with size $2^n-1$, length $(2^n-1)^2$, and correlation $3+2(2^n-1)$ \cite{gong}. 
However, most of the known constructions via finite fields made use of the multiplicative cyclic group of $\F_{2^n}$. It was often overlooked that all $2^n+1$ rational places including the place at infinity of the rational function field over $\F_{2^n}$ form a cyclic structure under an automorphism of order $2^n+1$. Recently, by using this cyclic structure,
an explicit construction of binary sequences of length $2^n+1$, size $2^n-1$ and correlation upper bounded by $\lfloor 2^{(n+2)/2}\rfloor$ via cyclotomic function fields over the finite field $\F_{2^n}$ was given \cite{JMX22}.
Note that all the above constructions are based on finite fields of even characteristics.

There are also some constructions of binary sequences which are based on finite fields of odd characteristics \cite{Lempel,Legendre,Pat,Rushanan,Su, Jin21}.
Paterson proposed a family of pseudorandom binary sequences based on Hadamard difference sets and maximum distance separable (MDS) codes for length $p^2$ where $p\equiv 3\ (\text{mod } 4)$ is a prime \cite{Pat}. 
In 2006, a family of binary sequences of length $p$ which is an odd prime, family size $(p-1)/2$, and correlation bounded by $5+2\sqrt{p}$ were given in \cite{Rushanan}, now known as Weil sequences.
The idea of constructing Weil sequences is to derive sequences from single quadratic residue based on Legendre sequences using a shift-and-add construction.
In 2010, by combining the $p$-periodic Legendre sequence and the $(q-1)$-periodic Sidelnikov sequence, Su {\it et al.} introduced a new sequence of length $p(q-1)$, called the Lengdre-Sidelnikov sequence \cite{Su}.
Recently, Jin {\it et al.} provided several new families of sequences with low correlations of length $q-1$ for any prime power $q$ by using the cyclic multiplication group $\F_q^*$ and sequences of  length $p$ for any odd prime $p$  from the cyclic additive group $\F_p$ in \cite{Jin21}. 
 The idea of constructing  binary sequences is  using  multiplicative quadratic character over finite fields of odd characteristics. 

Though various constructions of binary sequences with good parameters have been proposed, there are still limited choices for sequence lengths. For applications in different scenarios,  adding or deleting sequence values, usually destroys the good correlation properties.
Thus, there is still a need to construct binary sequences having low correlations with more flexible choices
of parameters. However, it appears to be an open and challenging problem to find good binary sequences with new parameters.

\subsection{Our main result and techniques}
In this paper, we provide an explicit construction of binary sequences with a low correlation of length $p^n+1$ via cyclotomic function fields over the finite field $\F_{p^n}$ for any odd prime $p$. The correlation of this family of binary sequences is upper bounded by $4+\lfloor 2\cdot p^{n/2}\rfloor$. To our best knowledge, this is the first construction of binary sequences with a low correlation of length type $p^n+1$. In Table I, we list many families of binary sequences with low correlations for comparison. 
 It turns out that our binary sequences have competitive parameters. 



In order to understand the main idea of this article better, we give a high-level description of our techniques.
In the language of function fields, we denote $P_\Ga$ by the zero of $x-\Ga$ for any $\Ga\in\F_q$ in the rational function field $\F_q(x)$.
We employ all rational places include the infinity place of the rational function field to obtain binary sequences with length $q+1$ as given in \cite{JMX20,JMX21,JMX22}.
The automorphism group of the rational function field $\F_q(u)$ over $\F_q$ is isomorphic to the projective general linear group $\PGL_2(\F_q)$ and  there is an automorphism $\Gs$ with order $q+1$ such that $\{\Gs^j(P_0)\}_{j=0}^q$ consists of all $q+1$ rational places of $\F_q(u)$ \cite{HKT08,JMX20}.
Our binary sequences are constructed via the quadratic multiplicative character of $\F_q$ and evaluations of carefully chosen rational functions  $z_1,z_2,\cdots,z_S\in E=\F_q(u)$ on such cyclically ordered rational places $\{\Gs^j(P_0)\}_{j=0}^q$.
The correlation of this family of binary sequences is converted to determining the upper bound of the number of rational places of Kummer extension $E_{i}=E(y)$ with $y^2=z_i\cdot \s^{-t}(z_i)$ for $1\le i \le S$ and $1\le t\le q$ and Kummer extension $E_{i,j}=E(y)$ with $y^2=z_i\cdot \s^{-t}(z_j)$ for $1\le i\neq j\le S$ and $0\le t\le q$.

In order to obtain binary sequences with a low correlation, the number of rational places of $E_i$ or $E_{i,j}$ can not be large. 
Hence, $z_i\cdot \Gs^{-t}(z_j)$ can not be elements in $\F_q\cdot E^2$ except the case $t=0$ and $i=j$, since there are much more rational places in the constant field extensions. To overcome this problem, we introduce an equivalence relation and choose at most one element in each equivalence class. Moreover, the genus of $E_i$ or $E_{i,j}$ can not be large, we need to choose functions $z_i$ such that $z_i$ and $\Gs^{-t}(z_j)$ have the same pole. In particular, $z_i$ can be chosen from a Riemann-Roch space associated to some place which is invariant under the automorphism $\Gs$. 

In particular, we employ the cyclotomic function fields over $\F_q(x)$ with modulus $p(x)$ which is a primitive quadratic irreducible polynomial to construct such an explicit family of binary sequences with a low correlation. The Galois group of this cyclotomic function field over $\F_q(x)$ is a cyclic group of order $q^2-1$. There is a unique subgroup of order $q-1$, its fixed subfield is the function field $E$ and the Galois group of $E$ over $\F_q(x)$ is a cyclic group of $q+1$ which is generated by $\sigma$. 
Let $Q$ be the unique totally ramified place of $E$ lying over $p(x)$ and $\mL(Q)$ be the Riemann-Roch space of $Q$. Hence, $\sigma(Q)=Q$ and rational functions $z_1,z_2,\cdots,z_S$ can be chosen as representative elements of some equivalence classes of an equivalence relation on $\mL(Q)\setminus \{0\}$. This family of binary sequences with a low correlation via cyclotomic function fields over finite fields with odd characteristics are distinct from the even case given in \cite{JMX22}, since there are non-negligible differences among the determination of cyclic automorphisms $\sigma$, the equivalence relations on $\mL(Q)\setminus \{0\}$, representative elements of equivalence classes, the construction of binary sequences and the estimation of their correlations.

\begin{table*}
	\setlength{\abovecaptionskip}{0pt}%
	\setlength{\belowcaptionskip}{10pt}%
	\caption{PARAMETERS OF SEQUENCE FAMILIES}
	\centering
{
	\begin{tabular}{@{}cccc@{}}
		\toprule
		Sequence                             & Length N                                & Family Size                           &  Bound of Correlation                       \\ \midrule
		\multicolumn{1}{|c|}{Gold(odd) \cite{Span}}      & \multicolumn{1}{c|}{$2^n-1$, $n$ odd}                   & \multicolumn{1}{c|}{$N+2$}             & \multicolumn{1}{c|}{$1+\sqrt{2}\sqrt{N+1}$} \\ \midrule
		\multicolumn{1}{|c|}{Gold(even) \cite{Span}}     & \multicolumn{1}{c|}{$2^n-1,n=4k+2$}            & \multicolumn{1}{c|}{$N+2$}             & \multicolumn{1}{c|}{$1+2\sqrt{N+1}$}        \\ \midrule
		\multicolumn{1}{|c|}{Kasami(small) \cite{Kas}}  & \multicolumn{1}{c|}{$2^n-1, n \text{ even}$}            & \multicolumn{1}{c|}{$\sqrt{N+1}$}      & \multicolumn{1}{c|}{$1+\sqrt{N+1}$}         \\ \midrule
		\multicolumn{1}{|c|}{Kasami(large) \cite{Kas}}  & \multicolumn{1}{c|}{$2^n-1,n=4k+2$}            & \multicolumn{1}{c|}{$(N+2)\sqrt{N+1}$} & \multicolumn{1}{c|}{$1+2\sqrt{N+1}$}        \\ \midrule
		\multicolumn{1}{|c|}{Bent \cite{Bent}}           & \multicolumn{1}{c|}{$2^n-1, n=4k$}                & \multicolumn{1}{c|}{$\sqrt{N+1}$}         & \multicolumn{1}{c|}{$1+\sqrt{N+1}$}          \\ \midrule
		\multicolumn{1}{|c|}{No \cite{No}}           & \multicolumn{1}{c|}{$2^n-1, n \text{ even}$}                & \multicolumn{1}{c|}{$\sqrt{N+1}$}         & \multicolumn{1}{c|}{$1+\sqrt{N+1}$}          \\ \midrule
		\multicolumn{1}{|c|}{Trace Norm \cite{Klap}}           & \multicolumn{1}{c|}{$2^n-1, n \text{ even}$}                & \multicolumn{1}{c|}{$\sqrt{N+1}$}         & \multicolumn{1}{c|}{$1+\sqrt{N+1}$}          \\ \midrule
		\multicolumn{1}{|c|}{Tang et al. \cite{Tang3}}  & \multicolumn{1}{c|}{$2(2^n-1),n$ odd }            & \multicolumn{1}{c|}{$\frac{N}{2}+1$} & \multicolumn{1}{c|}{$2+\sqrt{N+2}$}        \\ \midrule
		\multicolumn{1}{|c|}{Gong \cite{gong}}  & \multicolumn{1}{c|}{$(2^n-1)^2$, $2^n-1$ prime}            & \multicolumn{1}{c|}{$\sqrt{N}$ }& \multicolumn{1}{c|}{$3+2\sqrt{N}$}        \\ \midrule
	\multicolumn{1}{|c|}{Jin et al. \cite{JMX22}} &\multicolumn{1}{c|}{{\bf $2^n+1$}}  &\multicolumn{1}{c|}{{\bf $N-2$}}         & \multicolumn{1}{c|}{\bf {$2\sqrt{N-1}$}}          \\\midrule

		\multicolumn{1}{|c|}{Paterson \cite{Pat}}       & \multicolumn{1}{c|}{$p^2$, $p$ prime, $p\equiv 3(\text{mod } 4)$}    & \multicolumn{1}{c|}{$N$}               & \multicolumn{1}{c|}{$5+4\sqrt{N}$}          \\ \midrule
		\multicolumn{1}{|c|}{Paterson \cite{Pat}}  & \multicolumn{1}{c|}{$p^2$, $p$ prime, $p\equiv 3(\text{mod } 4)$}    & \multicolumn{1}{c|}{$\sqrt{N}+1$}      & \multicolumn{1}{c|}{$3+2\sqrt{N}$}          \\ \midrule
		\multicolumn{1}{|c|}{Weil \cite{Rushanan}}           & \multicolumn{1}{c|}{$p$, odd prime}                & \multicolumn{1}{c|}{$(N-1)/2$}         & \multicolumn{1}{c|}{$5+2\sqrt{N}$}          \\ \midrule

		\multicolumn{1}{|c|}{Jin et al. \cite{Jin21}} & \multicolumn{1}{c|}{$p-1$, $p\ge17$ odd prime power} & \multicolumn{1}{c|}{$N+3$}               & \multicolumn{1}{c|}{$6+2\sqrt{N+1}$}          \\ \midrule
		\multicolumn{1}{|c|}{Jin et al. \cite{Jin21}} & \multicolumn{1}{c|}{$p-1$, $p\ge11$ odd prime power} & \multicolumn{1}{c|}{$N/2$}         & \multicolumn{1}{c|}{$2+2\sqrt{N+1}$}          \\ \midrule
		\multicolumn{1}{|c|}{Jin et al. \cite{Jin21}} & \multicolumn{1}{c|}{$p$, $p\ge 17$ odd prime} & \multicolumn{1}{c|}{$N$}               & \multicolumn{1}{c|}{$5+2\sqrt{N}$}          \\ \midrule
		\multicolumn{1}{|c|}{Jin et al. \cite{Jin21}} & \multicolumn{1}{c|}{$p$, $p \ge 11$ odd prime} & \multicolumn{1}{c|}{$(N-1)/2$}         & \multicolumn{1}{c|}{$1+2\sqrt{N}$}          \\ \midrule
			\multicolumn{1}{|c|}{{\bf Our construction}} & \multicolumn{1}{c|}{ {$p^n+1$, $p$ odd prime}} & \multicolumn{1}{c|}{{\bf $N-3$}}         & \multicolumn{1}{c|}{\bf {$4+2\sqrt{N-1}$}}          \\
\bottomrule
	\end{tabular}
}
\end{table*}

\subsection{Organization of this paper}
In Section \ref{sec:2}, we provide preliminaries on binary sequences, rational function fields, extension theory of function fields, cyclotomic function fields and Kummer extensions.
In Section \ref{sec:3}, we present a theoretical construction of binary sequences with a low correlation via cyclotomic function fields over finite fields with odd characteristics. In Section \ref{sec:4}, we give an algorithm to generate such a family of binary sequences and provide many numerical results with the help of software Sage.

\section{Preliminaries}\label{sec:2}
In this section, we present some preliminaries on definitions and basic theory of binary sequences and their correlation, rational function fields, extension theory of function fields,  cyclotomic function fields and Kummer extensions.

\subsection{Binary sequences and their correlation}
Let $N$ be a positive integer. Let $\mS$ be a family of binary sequences with length $N$. For every sequence $\bs=(s_0,s_1,\dots,s_{N-1})\in\mS$ with $s_i\in\{1,-1\}$, we define the autocorrelation of $\bs$ at delay $t$ for $1\le t\le N-1$ by
\begin{equation}\label{eq:2.1}
A_t(\bs):=\sum_{i=0}^{N-1}s_is_{i+t},
\end{equation}
where $i+t$ means the least non-negative integer after taking modulo $N$. Consider two distinct binary sequences $\bu=(u_0,u_1,\dots,u_{N-1})$ and $ \bv=(v_0,v_1,\dots,v_{N-1})$ in $\mS$, we define the cross-correlation of $\bu$ and $\bv$ at delay $0\le t\le N-1$ by
\begin{equation}\label{eq:2.2}
C_t(\bu,\bv):=\sum_{i=0}^{N-1}u_iv_{i+t}.
\end{equation}
The correlation of the family of sequences $\mS$ is defined by 
\begin{equation}\label{eq:2.3}
Cor(\mS):=\max\left\{\max_{\bs\in\mS, 1\le t\le N-1}\{|A_t(\bs)|\},\max_{\bu\neq\bv\in\mS, 0\le t\le N-1}\{|C_t(\bu,\bv)|\}\right\}.
\end{equation}

\subsection{Rational function fields}
In this subsection, we introduce basic facts of the rational function field. The reader may refer to \cite{HKT08,JMX20,JMX21,St09} for more details.
Let $q$ be a prime power and $\F_q$ be the finite field with $q$ elements.
Let $K$ be the rational function field $\F_q(x)$, where $x$ is a transcendental element over $\F_q$.
Every finite place $P$ of $K$ corresponds to a monic irreducible polynomial $p(x)$ in $\F_q[x]$, and its degree $\deg(P)$ is equal to the degree of polynomial $p(x)$.
The pole of $x$ is called the infinity place of $K$ and denoted by $P_{\infty}$.
In fact, there are exactly $q+1$ rational places of $K$, i.e., the place $P_{\Ga}$ corresponding to $x-\Ga$ for each $\Ga\in \F_q$ and the infinite place $P_{\infty}$.

Let $P$ be a rational place of $K$ and let $\mathcal{O}_P$ be its valuation ring. For any $f\in \mathcal{O}_P$, $f(P)$ is defined to be the residue class of $f$ modulo $P$ in $\mathcal{O}_P/P\cong \F_q$; otherwise $f(P)=\infty$ for any $f\in K\setminus \mathcal{O}_P$.
If $f(x)=g(x)/h(x)\in K$ is written as a quotient of relatively prime polynomials, then the residue class map can be determined explicitly as follows
$$f(P_\Ga)=\begin{cases} g(\Ga)/h(\Ga) & \text{ if } h(\Ga)\neq 0 \\ \infty & \text{ if } h(\Ga)=0\end{cases} $$
for any $\Ga\in\F_q$. Moreover, if $\deg(g(x))<\deg(h(x))$, then $f(P_{\infty})=0.$

The automorphism group $\Aut(K/\F_q)$ of the rational function field $K$ over $\F_q$ is isomorphic to the projective general linear group $\PGL_2(\F_q)$.
Any automorphism $\Gs\in \Aut(K/\F_q)$ is uniquely determined by $\Gs(x)$ with the form
$$\Gs(x)=\frac{ax+b}{cx+d}$$
for some constants $a,b,c,d\in\F_q$ with $ad-bc\neq0$.
In particular, there exists an automorphism $\s\in  \Aut(K/\F_q)$ with order $q+1$ such that $\s$ acts cyclically on all rational places of $K$ from \cite{JMX20}.

\subsection{Extension theory of function fields}
Let $F/\F_q$ be an algebraic function field with genus $g$ over the full constant field $\F_q$. Let $\PP_F$ denote the set of places of $F$. Any place with degree one is called rational.  From the Serre bound \cite[Theorem 5.3.1]{St09}, the number $N(F)$ of rational places of $F$ is upper bounded by $$|N(F)-q-1|\le g \lfloor 2\sqrt{q}\rfloor. $$
Here $\lfloor x \rfloor$ stands for the integer part of $x\in \mathbb{R}$.

Let $\nu_P$ be the normalized discrete valuation of $F$ with respect to the place $P\in \PP_F$. The principal divisor of a nonzero element $z\in F$ is given by $(z):=\sum_{P\in \mathbb{P}_F}\nu_P(z)P$.
For a divisor $G$ of $F/\F_q$, the Riemann-Roch space associated to $G$ is defined by
\[\mL(G):=\{z\in F^*:\; (z)+G\ge 0\}\cup\{0\}.\]
If $\deg(G)\ge 2g-1$, then $\mL(G)$ is a vector space over $\F_q$ of dimension $\deg(G)-g+1$ from the Riemann-Roch theorem \cite[Theorem 1.5.17]{St09}.

Let $E/\F_q$ be a finite extension of function field $F/\F_q$. The Hurwitz genus formula \cite[Theorem 3.4.13]{St09} yields
$$2{g(E)}-2=[E:F]\cdot (2g(F)-2)+\deg \text{ Diff}(E/F),$$
where $\text{Diff}(E/F)$ stands for the different of $E/F$.
For $P\in \PP_F$ and $Q\in \PP_E$ with $Q|P$, let $d(Q|P), e(Q|P)$ be the different exponent and ramification index of $Q|P$, respectively. Then the different of $E/F$ can be given by
$\text{Diff}(E/F)=\sum_{Q\in\PP_E} d(Q|P) Q.$
If $p\nmid e(Q|P)$, then $d(Q|P)=e(Q|P)-1$ from Dedekind's Different Theorem \cite[Theorem 3.5.1]{St09}.

Let $\Aut(F/\F_q)$ denote the automorphism group of $F$ over $\F_q$, i.e.,
$\Aut(F/\F_q)=\{\Gs: F\rightarrow F |\; \Gs  \mbox{ is an } \F_q\mbox{-automorphism of } F\}.$
We can consider the group action of automorphism group $\Aut(F/\F_q)$ on the set of places of $F$.
From \cite[Lemma 1]{NX14}, we have the following results.

\begin{lemma}\label{lem:2.1}
For any automorphism $\s\in \Aut(F/\F_q)$, $P\in \mathbb{P}_F$ and $f\in F$, we have
\begin{itemize}
\item[(1)] $\deg(\s(P))=\deg(P)$;
\item[(2)] $\nu_{\s(P)}(\s(f))=\nu_P(f)$;
\item[(3)] $\s(f)(\s(P))=f(P)$ provided that $\nu_P(f)\ge 0$.
\end{itemize}
\end{lemma}

From Lemma \ref{lem:2.1}, we have $\sigma(\mathcal{L}(D)) = \mathcal{L}(\sigma(D))$ for any divisor $D$ of $F$. In particular, if $\s(P)=P$, then $\s(\mL(rP))=\mL(\s(rP))=\mL(rP)$ for any $r\in \mathbb{N}$.

\subsection{Cyclotomic function fields}
The basic theory of cyclotomic function fields was developed in the language of function fields by Hayes \cite{Ha74}.
 Let $x$ be an indeterminate over $\F_q$, $R$ be the polynomial ring $\F_q[x]$, $K$ be the rational function field $\F_q(x)$, and $K^{ac}$ be the algebraic closure of $K$. Let $\varphi$ be the endomorphism of $K^{ac}$ given by $$\varphi(z)=z^q+xz $$  for all $z\in K^{ac}$. Define a ring homomorphism
$$R\rightarrow \text{End}_{\mathbb{F}_q}(K^{ac}), f(x)\mapsto f(\varphi).$$
Then the $\F_q$-vector space of $K^{ac}$ is made into an $R$-module by introducing the following action of $R$ on $K^{ac}$, that is,
$$ z^{f(x)}=f(\varphi)(z)$$  for all $f(x)\in R$ and $z\in K^{ac}$. For a nonzero polynomial $M\in R$, we consider the set of $M$-torsion points of $K^{ac}$ defined by $$\Lambda_M=\{z\in K^{ac}| z^M=0\}.$$
The cyclotomic function field over $K$ with modulus $M$ is defined by the subfield of $K^{ac}$ generated over $K$ by all elements of $\Lambda_M$, and it is denoted by $K(\Lambda_M)$.  Let $p(x)=x^2+ax+b$ be an irreducible polynomial in $\mathbb{F}_q[x]$. In particular, we have the following facts from \cite{MXY16, JMX22}.

\begin{prop}\label{prop:2.2}
Let $p(x)=x^2+ax+b$ be an irreducible polynomial in $\mathbb{F}_q[x]$.
Let $F$ be the cyclotomic function field $K(\Lambda_{p(x)})$ with modulus $p(x)$ over $K$.
Then the following results hold:
\begin{itemize}
\item[\rm (i)] $[F:K]=q^2-1$ and $F= K(\Gl), \text{ where } \Gl^{q^2-1} +(x^q+x+a)\Gl^{q-1} +x^2+ax+b = 0.$
\item[\rm (ii)]   There is a unique place of $F$ lying over $p(x)$ which is totally ramified in $F/K$.
\item[\rm (iii)] The infinite place $\infty$ of $K$ splits into $q+1$ rational places, each with ramification index $q-1$ in the extension $F/K$.
\item[\rm (iv)] All other places of $K$ except $p(x)$ and $\infty$ are unramified in $F/K$.
\item[\rm (v)] The Galois group of $F$ over $K$ is $\Gal(F/K)\cong (\F_q[x]/(p(x)))^*$.  Moreover, the automorphism $\sigma_f\in \Gal(F/K)$ associated to $\overline{f}\in (\mathbb{F}_q[x]/(p(x)))^*$ is determined by $\sigma_f(\lambda)=\lambda^f$.
\end{itemize}
\end{prop}

\subsection{Kummer extensions}
The theory of Kummer extension of function fields can be summarized as follows from \cite[Proposition 3.7.3 and Corollary 3.7.4]{St09}.
\begin{prop}\label{prop:2.3}
Let $K$ be the rational function field $\F_q(x)$ and $n$ be a positive divisor of $q-1$. Suppose that $u\in K$ is an element satisfying $u\neq \omega^d \text{ for all } \omega\in F \text{ and } d|n$ with $d>1.$
Let $F$ be the Kummer extension over $K$ defined by \[F=K(y) \text{ with } y^n=u.\] Then we have:
\begin{itemize}
\item[(a)] The polynomial $\phi(T)=T^n-u$ is the minimal polynomial of $y$ over $K$. The extension $F/K$ is a cyclic extension of degree $n$, and the automorphisms of $F/K$ are given by $\sigma(y)=\zeta y$, where $\zeta\in \F_q$ is an $n$-th root of unity.
\item[(b)] Let $Q\in \PP_F$ be an extension of $P\in \mathbb{P}_K$. Let $r_P$ be the greatest common divisor of $n$ and $\nu_P(u)$, i.e., $r_P=\gcd(n,\nu_P(u))$. Then one has
\[e(Q|P)=\frac{n}{r_P}\quad \text{and} \quad d(Q|P)=\frac{n}{r_P}-1.\]
\item[(c)] Assume that there is a place $R\in \PP_F$ such that $\gcd(\nu_Q(u),n)=1$. Then $\F_q$ is the full constant field of $F$.
\end{itemize}
\end{prop}

\section{Binary sequences with a low correlation of length $q+1$}\label{sec:3}
Let $p$ be an odd prime, $m$ be a positive integer, $q=p^m$ be a prime power and $\F_q$ be the finite field with $q$ elements.
In this section, we will construct a family of binary sequences with a low correlation of length $q+1$ via cyclotomic function fields over the finite field $\F_q$ with an odd characteristic $p$.

\subsection{Fixed subfields of cyclotomic function fields}\label{subsec:3.1}
Let $K$ be the rational function field $\F_q(x)$ defined over $\F_q$.
Let $p(x)=x^2+ax+b$ be a primitive irreducible polynomial in $\mathbb{F}_q[x]$.
Let $F$ be the cyclotomic function field $K(\Lambda_{p(x)})$ with modulus $p(x)$ over $K$.
From Proposition \ref{prop:2.2}, we have $F=K(\lambda)$ and the Galois group $\Gal(F/K)$ is an abelian group of order $q^2-1$.
Moreover, we have the following result which is similar as \cite[Proposition 3.2]{JMX22} for the case of odd characteristics.

\begin{lemma}\label{lem:3.1}
Let $K$ be the rational function field $\F_q(x)$.
Let $p(x)=x^2+ax+b$ be a primitive irreducible polynomial in $\mathbb{F}_q[x]$.
Let $F$ be the cyclotomic function field $K(\Lambda_{p(x)})$ with modulus $p(x)$ over $K$.
Then the Galois group $\Gal(F/K)$ is a cyclic group of order $q^2-1$.
\end{lemma}
\begin{proof}
Let $\eta$ be the $K$-automorphism of $F$ defined by $\eta(\lambda)=\lambda^x$ from Proposition \ref{prop:2.2}.
Then $\eta^i$ can be determined by $\eta^i(\lambda)=\lambda^{x^i}$ for any positive integer $i$.
Hence, $\eta^i=id$ if and only if $x^i\equiv 1\ (\text{mod } p(x))$.
Since $p(x)$ is primitive, the order of $\eta$ is $q^2-1$.
\end{proof}

From Lemma \ref{lem:3.1}, there exists a unique subgroup $G$ of $\Gal(F/K)$ with order $q-1$. Such a unique subgroup $G$ can be determined explicitly as follows.

\begin{prop}\label{prop:3.2}
Let $\eta$ be the generator of $\Gal(F/K)$ determined by $\eta(\Gl)=\Gl^x$ and let $\tau=\eta^{q+1}$ be an automorphism of $F$. Then the unique subgroup $G$ of  $\Gal(F/K)$ with order $q-1$ is the cyclic group generated by $\tau$, i.e., $$G=\langle \tau \rangle=\{\tau_c\in \Gal(F/K): \tau_c(\lambda)=c\lambda \text{ for } c\in \F_q^*\}.$$
\end{prop}
\begin{proof}
From the Eisenstein's irreducibility criterion \cite[Proposition 3.1.15]{St09}, we have $x^2+ax+b$ divides $x^q+x+a$, i.e., 
$$x^q\equiv -x-a \ (\text{mod } x^2+ax+b).$$
Let $\tau=\eta^{q+1}$. It is easy to verify that $\tau(\Gl)=\eta^{q+1}(\lambda)=\lambda^{x^{q+1}}=\lambda^{-x(x+a)}=\Gl^b=b\Gl$, since $x^{q+1}=x\cdot x^q\equiv -x(x+a)\equiv b\ (\text{mod } x^2+ax+b).$
From \cite[Lemma 3.17]{LN83}, the constant $b$ is a generator of $\F_q^*$.
Hence, the order of $\tau$ is $q-1$ and $G=\langle \tau \rangle=\{\tau_c\in \Gal(F/K): \tau_c(\lambda)=c\lambda \text{ for } c\in \F_q^*\}$ from Proposition \ref{prop:2.2}.
\end{proof}

Let $E$ be the fixed subfield of $F$ with respect to $G$, that is, $E=F^G=\{z\in F: \sigma(z)=z \text{ for any } \sigma\in G\}.$
From Galois theory, $E/K$ is an abelian extension of degree $q+1$. In particular, the fixed subfield $E$ can be characterized explicitly as follows.

\begin{prop}\label{prop:3.3}
Let $u=\Gl^x/x=\Gl^{q-1}+x$. Then we have $F=K(\Gl)=\F_q(u,\Gl)$ and $F/\F_q(u)$ is a Kummer extension determined by
$$\Gl^{q-1}=-\frac{u^2+au+b}{u^q-u}.$$
Moreover, $u$ and $x$ satisfy the following equation
$$\frac{u^{q+1}+au+b}{u^q-u}=x.$$
\end{prop}
\begin{proof}
From Proposition \ref{prop:2.2} or \cite[Section 4]{MXY16}, the cyclotomic function field $F=K(\Lambda_{p(x)})$ is given by $F=K(\Gl)$ with
$\Gl^{q^2-1} +(x^q+x+a)\Gl^{q-1} +x^2+ax+b = 0.$
Substituting $x$ with $u-\Gl^{q-1}$, one has
$ \Gl^{q^2-1}+[(u-\Gl^{q-1})^q+(u-\Gl^{q-1})+a]\Gl^{q-1}+(u-\Gl^{q-1})^2+a(u-\Gl^{q-1})+b=0.$
It follows that $$\Gl^{q-1}=-\frac{u^2+au+b}{u^q-u}.$$  Since $x=u-\Gl^{q-1}$, we have $F=K(\Gl)=\F_q(x,\Gl)=\F_q(u,\Gl)$ and
$$x=u-\Gl^{q-1}=u+\frac{u^2+au+b}{u^q-u}=\frac{u^{q+1}+au+b}{u^q-u}. $$
\end{proof}

\begin{prop}\label{prop:3.4}
Let $\eta$ be the generator of $\Gal(F/K)$ determined by $\eta(\Gl)=\Gl^x$, $\tau=\eta^{q+1}$ and $G=\langle \tau \rangle$. Then the fixed subfield of $F$ with respect to $G$ is $E=F^G=\F_q(u)$.
\end{prop}
\begin{proof}
It is easy to verify that $\tau(\Gl^{q-1})=(\tau(\Gl))^{q-1}=(b\Gl)^{q-1}=\Gl^{q-1}$. From Proposition \ref{prop:3.3}, we have $\F_q(u)=\F_q(u,x)=\F_q(x,x+\Gl^{q-1})=\F_q(x,\Gl^{q-1})\subseteq F^{\langle \tau\rangle}=F^G.$
From Galois theory, the degree of extension $F/F^G$ is $[F:F^G]=|G|=q-1$.
From the proof of Proposition  \ref{prop:3.3}, $F/\F_q(u)$ is a Kummer extension given by
$$\Gl^{q-1}=-\frac{u^2+au+b}{u^q-u}.$$
From Proposition \ref{prop:2.3}, we have $[F:\F_q(u)]=q-1$.
Hence, we have $E=F^G=\F_q(u)$.
\end{proof}

Let $Q$ be a place of $E$ lying over $p(x)$.  From Proposition \ref{prop:2.2}, $Q|p(x)$ is totally ramified in $E/K$ with ramification index $e(Q|p(x))=q+1$ and different exponent $d(Q|p(x))=q$.
From the Hurwitz genus formula \cite[Theorem 3.4.13]{St09}, we obtain
$$-2=2g(E)-2\ge (q+1)[2g(K)-2]+2q.$$
It follows that all other places of $E$ except $Q$ are unramified in $E/K$.
Moreover, the place $Q$ can be characterized similarly as \cite[Proposition 4.2]{JMX22}.

\begin{lemma}\label{lem:3.5}
Let $Q$ be the unique place of $E$ lying over $p(x)$. Then the place $Q$ corresponds to the monic quadratic irreducible polynomial $p(u)=u^2+au+b$.
\end{lemma}
\begin{proof}
From Proposition \ref{prop:2.2}, $p(x)$ is totally ramified in $F/\F_q(x)$.
Hence, $Q$ is totally ramified in  $F/\F_q(u)$ with $\deg(Q)=2$.
From Proposition \ref{prop:2.3} and Proposition \ref{prop:3.3}, the zeros of $u^2+au+b$, $1/u$ and $u-\alpha$ with $\alpha\in \F_q$ are all totally ramified in the extension $F/\F_q(u)$. However, the zero of $u^2+au+b$ is the unique totally ramified place of $\F_q(u)$ in the extension $F/\F_q(u)$ with degree $2$.
Hence, the place $Q$ corresponds to the monic quadratic polynomial $p(u)=u^2+au+b$.
\end{proof}

\begin{theorem}\label{thm:3.6}
Let $\eta$ be the generator of $\Gal(F/K)$ determined by $\eta(\Gl)=\Gl^x$ and let $\sigma=\eta^{q}|_E$.
Then $\sigma$ is a $K$-automorphism of $E$ with order $q+1$ and the automorphism $\s$ can be determined explicitly by
$$\sigma(u)=\frac{-b}{u+a}.$$
In particular, the Galois group of $E/K$ is a cyclic group given by $\Gal(E/K)=\langle \sigma \rangle.$
\end{theorem}
\begin{proof}
From Galois theory, we have $\Gal(E/K)\cong \Gal(F/K)/\Gal(F/E)= \langle \eta \rangle/ \langle \eta^{q+1} \rangle$.
Thus, the order of $\eta$ is $q+1$ in the group $\Gal(E/K)$. Since $\gcd(q,q+1)=1$, the order of $\sigma=\eta^{q}|_E$ is $q+1$.
It is easy to verify that $\sigma(u)=\s(\Gl^x/\Gl)=\Gl^{x^{q+1}}/\Gl^{x^{q}}$.
From the proof of Proposition \ref{prop:3.2}, we have $x^q\equiv -x-a\ (\text{mod } x^2+ax+b)$ and $x^{q+1}=x\cdot x^q\equiv x(-x-a)\equiv b\  (\text{mod } x^2+ax+b)$.
It follows that $\Gl^{x^q}=\Gl^{-x-a}$, $\Gl^{x^{q+1}}=\Gl^b$ and
\begin{align*}
\sigma(u)&=\frac{\Gl^{x^{q+1}}}{\Gl^{x^{q}}}=\frac{\Gl^{b}}{\Gl^{-x-a}}=\frac{b\Gl}{-\Gl^x-a\Gl}=\frac{-b}{u+a}\in E.
\end{align*}
Hence, $\sigma$ is a $K$-automorphism of $E$ and $\Gal(E/K)=\langle \sigma \rangle.$
\end{proof}

\subsection{An equivalence relation on $\mL(Q)\setminus \{0\}$}\label{subsec:3.2}
Let $Q$ be the place of $E$ with degree two corresponding to the irreducible polynomial $p(u)=u^2+au+b$ from Lemma \ref{lem:3.5}.
Since the degree of $Q$ is $\deg(Q)=2\ge 2g(E)-1=-1$, the dimension of Riemann-Roch space $\mL(Q)$ is $\ell(Q)=\deg(Q)-g(E)+1=3$ from the Riemann-Roch theorem. It is clear that $\F_q=\mL(0)\subseteq \mL(Q)$.
Furthermore, $\mL(Q)$ can be determined explicitly as
$$\mL(Q)=\left\{\frac{c_0+c_1u+c_2u^2}{u^2+au+b}\in E: c_i\in \F_q \text{ for } 0\le i\le 2\right\}.$$
Now we can define a relation $\sim$ on $\mL(Q)\setminus \{0\}$ as follows: for any $z_1,z_2\in \mL(Q)\setminus \{0\}$,
\[z_1\sim z_2 \Leftrightarrow \exists \tau\in \Gal(E/K) \text{ such that } z_1\cdot \tau(z_2)\in \F_q\cdot E^2.\]
Here $E^2$ stands for the set $\{z^2: z\in E\}$.

\begin{prop}\label{prop:3.7}
The relation $\sim$ defined as above is an equivalence relation on the set $\mL(Q)\setminus \{0\}$.
\end{prop}
\begin{proof}
It is clear that three axioms of an equivalence relation are satisfied:
\begin{itemize}
\item[(1)] {\bf Reflexivity:} For any $z\in \mL(Q)\setminus \{0\}$, $z\cdot id(z)=z\cdot z=z^2\in \F_q\cdot E^2$. Hence, we have $z\sim z$.
\item[(2)]  {\bf Symmetry:}  If $z_1\sim z_2$, then there exists $\tau\in \Gal(E/K)$ such that $z_1\cdot \tau(z_2)=\alpha\cdot v^2\in \F_q\cdot E^2$ for some $\alpha\in \F_q$ and $v\in E$.
It follows that $z_2\cdot \tau^{-1}(z_1)=\tau^{-1}(\alpha\cdot v^2)=\tau^{-1}(\alpha)\cdot \tau^{-1}( v^2)=\Ga\cdot (\tau^{-1}(v))^2\in \F_q\cdot E^2$. Hence, we have $z_2\sim z_1$.
\item[(3)] {\bf Transitivity:}  If $z_1\sim z_2$ and $z_2\sim z_3$, there exist $\tau_1$ and $\tau_2$ in $\Gal(E/K)$ such that $z_1\cdot \tau_1(z_2)=\alpha_1\cdot v_1^2\in \F_q\cdot E^2$ and $z_2\cdot \tau_2(z_3)=\alpha_2\cdot v_2^2\in \F_q\cdot E^2$ for some $\alpha_1,\alpha_2\in \F_q$.
It follows that $z_1\cdot \tau_1(z_2)\cdot \tau_1(z_2\cdot \tau_2(z_3))=z_1\cdot (\tau_1(z_2))^2\cdot \tau_1\tau_2(z_3)=\alpha_1\Ga_2\cdot v_1^2(\tau_1(v_2))^2\in \F_q\cdot E^2$, i.e., $z_1\cdot \tau_1\tau_2(z_3)\in \F_q\cdot E^2$. Hence, we have $z_1\sim z_3$.
 \end{itemize}
\end{proof}

For any element $z\in \mL(Q)\setminus \{0\}$, let $[z]$ denote the equivalence class $\{x\in V: x\sim z\}$ containing $z$.
In this subsection, we want to determine representative elements of equivalence classes of  $\mL(Q)\setminus \{0\}$ under the relation $\sim$. Since $\Ga \cdot z\in [z]$ for any $\Ga\in \F_q^*$, the relation $\sim$ induces an equivalence relation on the set $V:=(\mL(Q)\setminus \{0\})/\F_q^*$.
Hence, it will be sufficient to determine representative elements of equivalent classes of $V$.
Let $S_1$ be the set
$$S_1=\left\{\frac{u^2+cu+d}{u^2+au+b}\notin \F_q: u^2+cu+d \text{ is an irreducible polynomial in } \F_q[u]\right\},$$ let $S_2$ be the set
$$S_2=\left\{\frac{u-\Ga}{u^2+au+b}: \Ga\in \F_q\right\}\cup \left\{\frac{(u-\Ga)(u-\Gb)}{u^2+au+b}: \Ga\neq \Gb\in \F_q\right\},$$ and let $S_3$ be the set
$$S_3=\left\{\frac{1}{u^2+au+b}\right\}\cup \left\{\frac{(u-\Ga)^2}{u^2+au+b}: \Ga\in \F_q\right\}.$$
From \cite[Corollary 3.21]{LN83}, the number of monic irreducible polynomials of degree $2$ in $\F_q[u]$ is $(q^2-q)/2$. Hence,  the cardinality of $S_1$ is $|S_1|=(q^2-q)/2-1=(q-2)(q+1)/2.$
By choosing representatives with monic leading coefficients of numerators of elements in $V$, we can identify $V$ as $S_1\cup S_2 \cup S_3\cup \{1\}$.

It is easy to see that all elements in $S_3$ are equivalent to $1/(u^2+au+b)$. 
In order to determine the equivalence classes of the set $S_1$ under the relation $\sim$, we need to study the ramification behavior of places of $E$ with degree $2$ in the extension $E/K$. The following result can be found from \cite[Proposition 1.4.12]{NX01}.
\begin{lemma}\label{lem:3.8}
Let $F/K$ be an abelian extension of function fields and let $E$ be an intermediate field of $F/K$.
Assume that the place $P$ of $K$ is unramified in $F/K$. Then $P$ splits completely in $E/K$ if and only if the Artin symbol $[\frac{F/K}{P}]$ belongs to $\Gal(F/E)$.
\end{lemma}

Now let us study the ramification behavior of places of $K$ with degree $2$ in the abelian extension $E/K$.

\begin{prop}\label{prop:3.9}
There are exactly $(q-3)/2$ distinct places of $K$ with degree $2$ which split completely in $E$, and there exists a unique rational place of $K$ which splits into $(q+1)/2$ places of $E$ with degree $2$. All such places of $E$ with degree $2$ correspond to the numerators of elements in $S_1$.
\end{prop}
\begin{proof}
There are $(q^2-q)/2$ places of the rational function field $E$ over $\F_q$ with degree $2$ in total from \cite[Corollary 3.21]{LN83} and \cite[Proposition 1.2.1]{St09}.
Let $Q_i$ be all pairwise distinct places of $E$ with degree $2$ other than $Q$ and let $P_i$ be its restriction to $K$ for $1\le i\le (q-2)(q+1)/2$.
From Proposition \ref{prop:2.2}, $Q_i|P_i$ is unramified in $E/K$ and the Artin symbol of $P_i\in \PP_K$ in the abelian extension $F/K$ is given by $$\left[\frac{F/K}{P_i}\right](\Gl)=\Gl^{P_i}.$$
Since $E$ is the fixed subfield of $F$ with respect to $G$ from Proposition \ref{prop:3.2} and Proposition \ref{prop:3.4}, the Galois group of $F/E$ is $\Gal(F/E)=G=\{\tau_c\in \Gal(F/K): \tau_c(\lambda)=c\lambda \text{ for } c\in \F_q^*\}.$
 From Lemma \ref{lem:3.8}, $P_i$ splits completely in $E/K$ if and only if there exists an element $\delta_i\in \F_q^*$ such that
$$\left[\frac{F/K}{P_i}\right](\Gl)=\Gl^{P_i}=\Gl^{\delta_i}.$$
Hence, $P_i=x^2+ax+b+\delta_i$ must be a quadratic irreducible polynomial in $\F_q[x]$.
In fact, the quadratic polynomial $P_i$ is irreducible if and only if there doesn't exist $\Ga\in \F_q$ such that $\Ga^2+a\Ga+b+\delta_i=0$, i.e., for any $\Ga\in \F_q$, we have $$\delta_i\neq -(\Ga^2+a\Ga+b)=-\left(\Ga+\frac{a}{2}\right)^2+\frac{a^2}{4}-b\in -\F_q^2+\frac{a^2}{4}-b.$$ 
Hence, $P_i=x^2+ax+b+\delta_i$ with $\delta_i\in \F_q^*\setminus (-\F_q^2+a^2/4-b)$ are irreducible polynomials in $\F_q[x]$.
Since the cardinality of $-\F_q^2+a^2/4-b$ is $(q+1)/2$, there are exactly $q-1-(q+1)/2=(q-3)/2$ different choices of $\delta\in \F_q^*$ such that $P_i$ is irreducible, i.e., there are exactly $(q-3)/2$ places of $K$ with degree $2$ which split completely into $(q-3)(q+1)/2$ places of $E$ with degree $2$.

Let $P$ be a place of $K$ which is unramified in $F/K$ and let $R$ be any place of $F$ lying over $P$. The order of Artin symbol $\left[\frac{F/K}{P}\right]$ is $f(R|P)=\deg(R)/\deg(P)$.
Since there are $(q-2)(q+1)/2$ places of $E$ with degree $2$ except $Q$, the remaining $(q+1)/2$ places of $E$ with degree $2$ must lie over a rational place of $K$ from \cite[Theorem 3.7.2]{St09}.
\end{proof}

From Proposition \ref{prop:3.9}, $P_i=x^2+ax+b+\delta_i$ with $\delta_i\in \F_q^*\setminus (-\F_q^2+a^2/4-b)$ are distinct places of $K$ of degree $2$ which split completely in $E/K$. Let $Q_i=u^2+c_iu+d_i\in \F_q[u]$ be any place of $E$ lying over $P_i$ for $1\le i\le (q-3)/2$.
Now we try to characterize the equivalence classes of $S_1$.
Let $R_1$ be the set defined by  
$$R_1:=\left\{z_i=\frac{u^2+c_iu+d_i}{u^2+au+b}: 1\le i\le  (q-3)/2\right\}.$$

\begin{prop}\label{prop:3.10}
For any $z_i\neq z_j\in R_1$, one has $z_i\cdot \tau(z_j)\notin \F_q\cdot E^2$ for any $\tau\in \Gal(E/K)$, i.e., $[z_i]$ are distinct equivalence classes of $\mL(Q)\setminus \{0\}$ for $1\le i\le (q-3)/2$.
\end{prop}
\begin{proof}
The principal divisor of $z_i=(u^2+c_iu+d_i)/(u^2+au+b)$ is given by $$(z_i)=\left(\frac{u^2+c_iu+d_i}{u^2+au+b}\right)=Q_i-Q.$$
From Lemma \ref{lem:2.1}, the principal divisor of $\tau(z_i)$ is $(\tau(z_i))=\tau((z_i))=\tau(Q_i-Q)=\tau(Q_i)-\tau(Q). $
Since $Q$ is totally ramified in $E/K$, one has $\tau(Q)=Q$ for $\tau\in \Gal(E/K)$.
For $z_i\neq z_j\in R_1$, we have $Q_i\neq Q_j$ and the principal divisor of $z_i\cdot \tau(z_j)$ is given by $$ (z_i\cdot \tau(z_j))=Q_i-Q+\tau(Q_j)-\tau(Q)=Q_i+\tau(Q_j)-2Q.$$
Since $Q_i$ and $Q_j$ lie over distinct places of $K$ with degree $2$ from Proposition \ref{prop:3.9}, we have $\tau(Q_j)\neq Q_i$ for any $\tau \in \Gal(E/K)$. Hence, we have  $z_i\cdot \tau(z_j)\notin \F_q\cdot E^2$.
\end{proof}

From Proposition \ref{prop:3.10}, $[z_i]$ are pairwise distinct equivalence classes of $\mL(Q)\setminus \{0\}$ for $1\le i\le (q-3)/2$.
In particular, their representative elements satisfy the following property.

\begin{prop}\label{prop:3.11}
For any $z_i\in R_1$ with $1\le i\le (q-3)/2$ and any automorphism $\tau\in \Gal(E/K)\setminus \{id\}$, one has $z_i\cdot \tau(z_i)\notin \F_q\cdot E^2$.
\end{prop}
\begin{proof}
The principal divisor of $z_i=(u^2+c_iu+d_i)/(u^2+au+b)$ is given by $$(z_i)=\left(\frac{u^2+c_iu+d_i}{u^2+au+b}\right)=Q_i-Q.$$
From Lemma \ref{lem:2.1} and Proposition \ref{prop:3.10}, the principal divisor of $z_i\cdot \tau(z_i)$ is given by $$\left(z_i\cdot \tau(z_i)\right)=Q_i-Q+\tau(Q_i)-\tau(Q)=Q_i+\tau(Q_i)-2Q.$$
Since $P_i$ splits completely in $E/K$ from Proposition \ref{prop:3.9}, we have $\tau(Q_i)\neq Q_i$. Hence, one has $z_i\cdot \tau(z_i)\notin \F_q\cdot E^2$.
\end{proof}

In the following, we want to determine the equivalence classes of $S_2$.
Since $E$ is a rational function field over $\F_q$, there exist exactly $q+1$ rational places which are places of $E$ lying over the infinity place $\infty$ of $K$ from Proposition \ref{prop:2.2}.
From Theorem \ref{thm:3.6}, the Galois group $\Gal(E/K)$ is a cyclic group generated by an automorphism $\sigma$ with order $q+1$.
Let $P_0$ be the zero of $u$ in $E$ and  $P_{\Gs^j}=\sigma^j(P_0)$ for $0\le j\le q$. From \cite[Theorem 3.7.1]{St09}, $P_{\Gs^j}$ are distinct rational places of $E$ for $0\le j\le q$.
In particular, we have $P_{\Gs^1}=\s(P_0)=P_\infty$ from Theorem \ref{thm:3.6}.
Let $\Ga_j$ be the element $u(P_{\Gs^j})$ in $\F_q$ for each $0\le j\neq 1\le q$. Then we have $P_{\Gs^j}=\Gs^j(P_0)=P_{u-\Ga_j}$. 
 It is easy to see that $\Ga_j$ are pairwise distinct elements of $\F_q$ for $0\le j\neq 1\le q$. Let $w_j=(u-\Ga_j)/(u^2+au+b)$ for $j\neq 1$ and let $R_2$ be a set defined by $$R_2:=\left\{w_j=\frac{u-\Ga_j}{u^2+au+b}: 2\le j\le (q+1)/2\right\}.$$

\begin{prop}\label{prop:3.12}
For $w_i\neq w_j\in R_2$, one has $w_i\cdot \s^{t}(w_j)\notin \F_q\cdot E^2$ for any $0\le t\le q$, i.e., $[w_j]$ are pairwise distinct equivalence classes of $\mL(Q)\setminus \{0\}$ for $2\le j\le (q+1)/2$.
\end{prop}
\begin{proof}
The principal divisor of $w_i$ is given by $$(w_i)=P_{\Gs^i}+P_\infty-Q=\Gs^i(P_0)+\Gs(P_0)-Q.$$
From Lemma \ref{lem:2.1}, the principal divisor of $\sigma^t(w_j)$ for each $0\le t\le q$ is
$$(\sigma^t(w_j))=\sigma^t(P_{\Gs^j})+\sigma^t(P_\infty)-\sigma^t(Q)=\Gs^{t+j}(P_0)+\Gs^{t+1}(P_0)-Q.$$
If $w_i \sim w_j$ for $2\le j\neq i\le (q+1)/2$, i.e., $w_i\cdot \sigma^t(w_j)\in \F_q\cdot E^2$ for some $0\le t\le q$, then we have 
$$\begin{cases} \Gs^{t+j}(P_0)=\Gs^i(P_0)\\  \Gs^{t+1}(P_0)=\Gs(P_0)\end{cases} \text{ or }\quad \begin{cases} \Gs^{t+j}(P_0)=\Gs(P_0)\\  \Gs^{t+1}(P_0)=\Gs^i(P_0)\end{cases}.$$
If $ \Gs^{t+1}(P_0)=\Gs(P_0)$, then we have $t=0$. It follows that $j=i$ which is impossible. 
If $ \Gs^{t+1}(P_0)=\Gs^i(P_0)$, then there exists an integer $t$ with $1\le t\le q$ such that $t+1=i$ and $t+j=1+q+1$. 
Hence, we obtain $i+j=q+3$ which is a contradiction. 
\end{proof}

Now we want to classify all elements of $S_2$ which are equivalent to $w_j$ for each $2\le j\le (q+1)/2$.
From the definition of the equivalence relation $\sim$, it is clear that $\Gs^{t}(w_j)$ is equivalent to $w_j$ for any $0\le t\le q$. 
From Lemma \ref{lem:2.1}, the principal divisor of $\sigma^t(w_j)$ for each $1\le t\le q$ is given by 
$(\sigma^t(w_j))=\Gs^{t+j}(P_0)+\Gs^{t+1}(P_0)-Q.$
If $\Gs^{t+j}(P_0)=P_\infty$, i.e., $t+j=1+(q+1)$, then $t+1=q+3-j$ and $w_j\sim w_{q+3-j}$ from the proof of Proposition \ref{prop:3.12}.
If $\Gs^{t+j}(P_0)\neq P_\infty$, i.e., $t+j\neq 1+(q+1)$, then $$w_j\sim \frac{(u-\Ga_{t+j})(u-\Ga_{t+1})}{u^2+au+b}. $$
Hence, the equivalence class $[w_j]$ contains at least $q+1$ elements in $S_2$ for each $2\le j\le (q+1)/2$.
For $j=(q+3)/2$, it is easy to verify that $w_{j}\cdot \sigma^{j-1}(w_{j})\in \F_q\cdot E^2$ and $[w_j]$ contains at least $(q+1)/2$ elements in $S_2$. The cardinality of $S_2$ is $q(q+1)/2$. 
Hence, $\{[w_i]: 2\le i\le (q+3)/2\}$ are all distinct equivalence classes of $S_2$. Furthermore, these representative elements in $R_2$ have the following property.

\begin{prop}\label{prop:3.13}
For any $w_j\in R_2$ with $2\le j\le (q+1)/2$, one has $w_j\cdot \s^{t}(w_j)\notin \F_q\cdot E^2$ for $1\le t\le q$, i.e., $w_j\cdot \tau(w_j)\notin  \F_q\cdot E^2$ for any $\tau\in \Gal(E/K)\setminus \{id\}$.
\end{prop}
\begin{proof}
The principal divisor of $w_j\cdot \s^t(w_j)$ for each $1\le t\le q$ is given by
\begin{align*}
(w_j\cdot \s^t(w_j))&=P_{\Gs^j}+P_\infty-Q+\s^t(P_{\Gs^j})+\s^t(P_\infty)-\Gs^t(Q)\\ &=\s^j(P_0)+\s(P_0)+\s^{t+j}(P_0)+\s^{t+1}(P_0)-2Q.
\end{align*}
Assume that $w_j\cdot \s^t(w_j)\in \F_q \cdot E^2$. Since $1\le t\le q$, we must have $\s^{t+1}(P_0)=\s^j(P_0)$ and $\s^{t+j}(P_0)=\s(P_0)$, i.e., $ t+1=j$ and $ t+j=1+q+1$.
Hence, we obtain $j=(q+3)/2$ which is a contradiction.
This completes the proof.
\end{proof}

\begin{theorem}\label{thm:3.14}
For any $z_i\in R_1, w_j\in R_2$, we have $z_i\cdot \s^{t}(w_j)\notin \F_q\cdot E^2$ and $w_j\cdot \s^{t}(z_i)\notin \F_q\cdot E^2$ for $0\le t\le q$, i.e., $[z_i]$ and $[w_j]$ are pairwise distinct equivalence classes of $\mL(Q)\setminus \{0\}$ for $1\le i\le (q-3)/2$ and $2\le j\le (q+1)/2$.
\end{theorem}
\begin{proof}
The principal divisor of $z_i\cdot \s^t(w_j)$ for each $0\le t\le q$ is given by
\begin{align*}
(z_i\cdot \s^t(w_j))&=Q_i-Q+\s^t(P_{\Gs^j})+\s^t(P_\infty)-Q\\ &=Q_i+\s^{t+j}(P_0)+\s^{t+1}(P_0)-2Q.
\end{align*}
It is easy to see that $z_i\cdot \s^{t}(w_j)\notin \F_q\cdot E^2$ for $0\le t\le q$.
From Proposition \ref{prop:3.10} and Proposition \ref{prop:3.12}, $[z_i]$ and $[w_j]$ are pairwise distinct equivalence classes of $\mL(Q)\setminus \{0\}$ for $1\le i\le (q-3)/2$ and $2\le j\le (q+1)/2$.
\end{proof}

\subsection{Binary sequences with a low correlation of length $q+1$}\label{subsec:3.3}
In this subsection, we provide a construction of binary sequences with a low correlation of length $q+1$ via cyclotomic function fields over finite fields with odd characteristics.
From Proposition \ref{prop:3.10}, let $ [z_1],[z_2],\cdots,[z_{\frac{q-3}{2}}]$ be pairwise distinct equivalence classes of $z_i\in R_1$.
From Proposition \ref{prop:3.12}, let $ [w_2],[w_3],\cdots,[w_{\frac{q+1}{2}}]$ be pairwise distinct equivalence classes of $w_j\in R_2$.
For simplicity, let $z_{\frac{q-5}{2}+j}=w_j$ for $2\le j\le (q+1)/2$. From Theorem \ref{thm:3.14}, $[z_i]$ are pairwise distinct equivalence classes of $\mL(Q)\setminus \{0\}$ for $1\le i\le q-2$.
Let $\eta$ be the quadratic character from $\F_q^*$ to $\mathbb{C}^*$, i.e.,
$$\eta(\Ga)=\begin{cases} 1 &\text{ if } \Ga \text{ is a square in } \F_q^*,\\ -1 & \text{ if } \Ga \text{ is a non-square in } \F_q^*.\end{cases}$$
Assume that $\eta(0)=1$. We can extend $\eta$ to a map from $\F_q$ to $\mathbb{C}^*$ by defining
$$\eta(\Ga)=\begin{cases} 1 &\text{ if } \Ga \text{ is a square in } \F_q,\\ -1 & \text{ if } \Ga \text{ is a non-square in } \F_q.\end{cases}$$
Let $P_0$ be the zero of $u$ in $E$ and  $P_{\Gs^j}=\sigma^j(P_0)$ for $0\le j\le q$.
For each equivalence class $[z_i]$ for $1\le i\le q-2$, we define a sequence $s_i$ as follows:
$$s_i=(s_{i,0},s_{i,1},\cdots,s_{i,q}) \text{ with } s_{i,j}=\eta(z_i(P_{\Gs^j})) \text{ for } 0\le j\le q.$$
In the following, we will show that this family of binary sequences $\{s_i: 1\le i\le q-2\}$ with length $q+1$ has a low correlation.

\begin{prop}\label{prop:3.15}
If $q=p^m$ is a power of odd prime, then the autocorrelation of $s_i$ with $1\le i\le q-2$ at delay $t$ for $1\le t\le q$ is upper bounded by $$|A_t(s_i)|\le 4+\lfloor 2\sqrt{q}\rfloor.$$
\end{prop}
\begin{proof}
The autocorrelation of $s_i$ at delay $t$ for $1\le t\le q$ is given by
\begin{align*}
A_t(s_i)&= \sum_{j=0}^{q} s_{i,j}s_{i,j+t}=  \sum_{j=0}^q \eta(z_i(P_{\Gs^j}))\cdot \eta(z_i(P_{\Gs^{j+t}}))\\ &= \sum_{j=0}^q \eta(z_i(P_{\Gs^j}))\cdot \eta((\sigma^{-t}(z_i))(P_{\Gs^j})).
\end{align*}
Let $Z_i=\{0\le j\le q: (z_i\sigma^{-t}(z_i))(P_{\Gs^j})=0\}$. It is easy to see that $|Z_i|\le 4$, since $z_i\in \mL(Q)$ and $z_i\sigma^{-t}(z_i)\in \mL(2Q)$ from Proposition \ref{prop:3.11} and Proposition \ref{prop:3.13}. Thus, we have
$$A_t(s_i)= \sum_{j\in Z_i} \eta(z_i(P_{\Gs^j}))\cdot \eta(z_i(P_{\Gs^{j+t}}))+\sum_{j\notin Z_i} \eta((z_i\sigma^{-t}(z_i))(P_{\Gs^j})).$$
For any $z_i\in R_1\cup R_2$, we have $z_i\sigma^{-t}(z_i)\notin \F_q\cdot E^2$ for $1\le t\le q$ from Proposition \ref{prop:3.11} and Proposition \ref{prop:3.13}.
Consider the Kummer extension $E_i/E$ given by $E_i=E(y)$ with
$$y^2=z_i\cdot \sigma^{-t}(z_i).$$
If $z_i\in R_1$, then $(z_i\cdot \sigma^{-t}(z_i))=Q_i+\Gs^{-t}(Q_i)-2Q$  from Proposition \ref{prop:3.11}. Hence, there are two places of $E$ with degree $2$ which are totally ramified in $E_i/E$  from Proposition \ref{prop:2.3}.
If $z_i\in R_2$, then $(z_i\cdot \sigma^{-t}(z_i))=\Gs^i(P_0)+\Gs(P_0)+\Gs^{i-t}(P_0)+\Gs^{1-t}(P_0)-2Q$  from Proposition \ref{prop:3.13}. Hence, there are at most four rational places of $E$ which are totally ramified in $E_i/E$.
In both cases, we have $\deg \text{Diff}(E_i/E)\le 4$. Furthermore, $\F_q$ is the full constant field of $E_i$ from Proposition \ref{prop:2.3}.
The Hurwitz genus formula yields
$$2g(E_i)-2= 2\cdot (2g(E)-2)+\deg \text{Diff}(E_i/E).$$
Hence, the genus of $E_i$ is at most $1$ for each $1\le i\le q-2$.

Let $N_1$ denote the cardinality of the set $\{j\notin Z_i: \eta((z_i\cdot \sigma^{-t}(z_i))(P_{\Gs^j}))=1\}$ and let $N_{-1}$ denote the cardinality of the set $\{j\notin Z_i: \eta((z_i\cdot \sigma^{-t}(z_i))(P_{\Gs^j}))=-1\}$. It is clear that $$N_1+N_{-1}=q+1-|Z_i|.$$
From \cite[Theorem 3.3.7]{St09}, the number of rational places of $E_i$ is $N(E_i)=2N_1+|Z_i|$. From the Serre bound,  $N(E_i)$ is bounded by
$$q+1-\lfloor 2\sqrt{q}\rfloor\le N(E_i)=2N_1+|Z_i|\le q+1+\lfloor 2\sqrt{q}\rfloor.$$
Hence, we have $-\lfloor 2\sqrt{q}\rfloor\le N_1-N_{-1}\le \lfloor 2\sqrt{q}\rfloor$ and
$$|A_t(s_i)|\le |Z_i|+|N_1-N_{-1}|\le 4+\lfloor 2\sqrt{q}\rfloor.$$
\end{proof}

\begin{prop}\label{prop:3.16}
If $q=p^m$ is a power of odd prime, then the cross-correlation of $s_i$ and $s_j$ with $1\le i\neq j\le q-2$ at delay $t$ for $0\le t\le q$ is upper bounded by $$|C_t(s_i,s_j)|\le 4+\lfloor 2\sqrt{q}\rfloor.$$
\end{prop}
\begin{proof}
For two distinct sequences $s_i$ and $s_j$ in $\mS$, the cross-correlation of $s_i$ and $s_j$ with $1\le i\neq j\le q-2$ at delay $t$ for $0\le t\le q$ is given by
\begin{align*}C_t(s_i,s_j)&= \sum_{k=0}^{q} s_{i,k}s_{j,k+t} =  \sum_{k=0}^q \eta(z_i(P_{\Gs^k}))\cdot \eta(z_j(P_{\Gs^{k+t}})) \\
&=\sum_{k=0}^q \eta(z_i(P_{\Gs^k}))\cdot \eta((\sigma^{-t}(z_j))(P_{\Gs^k})).\end{align*}
Let $Z_{i,j}=\{0\le k\le q: (z_i\sigma^{-t}(z_j))(P_{\Gs^k})=0\}$ be the zeros of $z_i\sigma^{-t}(z_j)$ with degree one. It is easy to see that $|Z_{i,j}|\le 4$, since $z_i\sigma^{-t}(z_j)\in \mL(2Q)$ for any $z_i,z_j\in \mL(Q)$. Thus, we have
$$C_t(s_i,s_j) =  \sum_{j\in Z_{i,j}} \eta(z_i(P_{\Gs^k}))\cdot \eta((\sigma^{-t}(z_j))(P_{\Gs^k}))+\sum_{j\notin Z_{i,j}} \eta((z_i\sigma^{-t}(z_j))(P_{\Gs^k})).$$
For $z_i\neq z_j\in R_1\cup R_2$, we have $z_i\sigma^{-t}(z_j)\notin \F_q\cdot E^2$ from Proposition \ref{prop:3.10}, Proposition \ref{prop:3.12} and Theorem \ref{thm:3.14}.
Let us consider the Kummer extension $E_{i,j}=E(y)$ with
$$y^2=z_i\cdot \sigma^{-t}(z_j).$$
If $z_i\neq z_j\in R_1$, then there are two places of $E$ with degree $2$ which are totally ramified in $E_{i,j}/E$ from Proposition \ref{prop:3.10}.
If $z_i\neq z_j\in R_2$, then there are at most four rational places of $E$ which are totally ramified in $E_{i,j}/E$ from  Proposition \ref{prop:3.12}.
If $z_i\in R_1, z_j\in R_2$ or $z_i\in R_2, z_j\in R_1$, then there are a place of degree $2$ and two rational places of $E$ which are totally ramified in $E_{i,j}/K$ from Theorem \ref{thm:3.14}.
In these cases, we have $\deg \text{Diff}(E_{i,j}/E)\le 4$.
From Proposition \ref{prop:2.3}, $\F_q$ is the full constant field of $E_{i,j}$ as well.
The Hurwitz genus formula yields
$$2g(E_{i,j})-2= 2\cdot (2g(E)-2)+\deg \text{Diff}(E_{i,j}/E).$$
Hence, the genus of $E_{i,j}$ is at most $1$ for any $1\le i\neq j\le q-2$.

Let $N_1$ be the number of the set $\{k\notin Z_{i,j}: \eta((z_i\cdot \sigma^{-t}(z_j))(P_{\Gs^k})))=1\}$ and
let $N_{-1}$ be the number of the set $\{k\notin Z_{i,j}: \eta((z_i\cdot \sigma^{-t}(z_j))(P_{\Gs^k}))=-1\}$.
It is clear that $$N_1+N_{-1}=q+1-|Z_{i,j}|.$$
From \cite[Theorem 3.3.7]{St09} and the Serre bound, the number of rational places of $E_{i,j}$ is bounded by
$$q+1-\lfloor 2\sqrt{q}\rfloor\le N(E_{i,j})=2N_1+|Z_{i,j}|\le q+1+\lfloor 2\sqrt{q}\rfloor.$$
Hence, we have $-\lfloor 2\sqrt{q}\rfloor\le N_1-N_{-1}\le \lfloor 2\sqrt{q}\rfloor$ and
$$|C_t(s_i,s_j)|\le |Z_{i,j}|+|N_1-N_{-1}|\le  4+\lfloor 2\sqrt{q}\rfloor.$$
The proof is completed.
\end{proof}

\begin{theorem}\label{thm:3.17}
If $q=p^m$ is a power of odd prime, then there exists a family of binary sequences $\mS=\{s_i: 1\le i\le q-2\}$ of length $q+1$ with a correlation upper bounded by $$\text{Cor}(\mS)\le 4+\lfloor 2\sqrt{q}\rfloor.$$
\end{theorem}
\begin{proof}
This theorem follows immediately from Proposition \ref{prop:3.15} and Proposition \ref{prop:3.16}.
\end{proof}

\section{Algorithm and numerical results}\label{sec:4}
The previous section provides a theoretical construction of binary sequences with a low correlation via cyclotomic function fields over finite fields with odd characteristics.
In fact, such a family of binary sequences constructed in Section \ref{sec:3} can be realized explicitly.  In this section, we provide an explicit construction of binary sequences via explicit automorphisms and compute some examples for finite fields of small sizes with the help of the software Sage.

From Theorem \ref{thm:3.6}, the Galois group  $\Gal(E/K)$ is a cyclic group generated by $\Gs$. From  \cite[Proposition 4.4]{JMX22}, all automorphisms $\s^j$ can be calculated explicitly by the following recursive relations for $0\le j\le q$.

\begin{lemma}\label{lem:4.1}
Let $a_0=1, b_0=0,c_0=0,d_0=1$ and $a_1=0, b_1=-b, c_1=1,d_1=a$.
Let $\s^j$ be the automorphism of $\F_q(u)$ determined by
$$\s^j(u)=\frac{a_ju+b_j}{c_ju+d_j}.$$
Then $a_j,b_j,c_j,d_j$ with $0\le j\le q$ can be obtained from the following recursive equations
$$\begin{cases}
a_{j+1}=a_1a_j+c_1b_j=b_j,\\
b_{j+1}=b_1a_j+d_1b_j=-ba_j+a b_j,\\
c_{j+1}=a_1c_j+c_1d_j=d_j,\\
d_{j+1}=b_1c_j+d_1d_j=-bc_j+a d_j.
\end{cases}$$
\end{lemma}

\begin{prop}\label{prop:4.2}
Let $P_0$ be the zero of $u$ in $E=\F_q(u)$, $\Gs$ be the $\F_q$-automorphism of $E$ determined by $\Gs(u)=-b/(u+a)$ and  $P_{\Gs^j}=\Gs^j(P_0)$.
Then  $P_{\Gs^1}=P_\infty$  and $P_{\Gs^j}$ corresponds to the linear polynomial  $u+a_j^{-1}b_j$ for each $0\le j\neq 1 \le q$.
\end{prop}
\begin{proof}
Since $\Gs(u)=-b/(u+a)$, we have $P_{\Gs^1}=\Gs(P_0)=P_\infty$.
From Lemma \ref{lem:4.1}, the automorphism $\Gs^j$ is determined by $\s^j(u)=(a_ju+b_j)/(c_ju+d_j).$
If $a_j\neq 0$, i.e., $j\neq 1$, then $P_{\Gs^j}$ corresponds to the linear polynomial $u+a_j^{-1}b_j$.
\end{proof}

Let $[z_1],[z_2],\cdots,[z_{q-2}]$ be the distinct equivalence classes of $\mL(Q)\setminus \{0\}$ under the equivalence relation $\sim$ in Subsection \ref{subsec:3.2}.
From the theory of rational function fields, we have $z_i(P_{\Gs^j})=z_i(u)|_{u=-b_j/a_j}$ if $0\le j\neq 1\le q$; otherwise, $z_i(P_{\Gs^1})=z_i(P_\infty)=1$ for $1\le i\le (q-3)/2$ and $z_i(P_{\Gs^1})=z_i(P_\infty)=0$ for $(q-1)/2\le i\le q-2$.
Now, let us provide an explicit construction of binary sequences with a low correlation via cyclotomic function fields over finite fields with odd characteristics obtained from Theorem \ref{thm:3.17}.
Such an explicit construction of binary sequences can be presented as follows.

\begin{center}
Construction of binary sequences with a low correlation
\end{center}
\begin{itemize}
\item Step 1: Input an odd prime $p$ and a prime power $q=p^m$.

\item Step 2: Find a primitive quadratic irreducible polynomial $p(x)=x^2+ax+b\in \F_q[x]$.

\item Step 3: Let $\s$ be an $\F_q$-automorphism of $\F_q(u)$ determined by  $\sigma(u)=-b/(u+a).$
Let $a_0=d_0=1,b_0=c_0=0$. For each $0\le j\le q$, calculate the explicit formula $\s^j(u)=(a_ju+b_j)/(c_ju+d_j)$ via the following recursive equations:
$$\begin{cases}
a_{j+1}=b_j,\\
b_{j+1}=-ba_j+a b_j,\\
c_{j+1}=d_j,\\
d_{j+1}=-bc_j+a d_j.
\end{cases}$$

\item Step 4:
Determine a quadratic irreducible polynomial $u^2+c_i u+d_i$ lying over the place $P_i=x^2+ax+b+\delta_i$ with $\delta_i\in \F_q^*\setminus (-\F_q^2+a^2/4-b)$ in the extension $E/K$ for each $1\le i\le (q-3)/2$. Define $z_i=(u^2+c_i u+d_i)/(u^2+au+b)$ for $1\le i\le (q-3)/2$.

\item Step 5: Let $\Ga_j=u(P_{\Gs^j})=-b_j/a_j$ for $2\le j\le (q+1)/2$. For $2\le j\le (q+1)/2$,  $$z_{\frac{q-5}{2}+j}=\frac{u-\Ga_j}{u^2+au+b}.$$

\item Step 6: Let $\eta$ be a map from $\F_q$ to $\mathbb{C}^*$ defined by
$$\eta(\Ga)=\begin{cases} 1 &\text{ if } \Ga \text{ is a square in } \F_q,\\ -1 & \text{ if } \Ga \text{ is a non-square in } \F_q.\end{cases}$$
Output a family of sequences $\mS=\{s_i: 1\le i\le q-2\}$ defined by $$s_i=(s_{i,0},s_{i,1},\cdots,s_{i,q}) \text{ with } s_{i,j}=\eta(z_i(P_{\Gs^j})) \text{ for } 0\le j\le q.$$

\item Step 7: 
Output the correlation $Cor(\mS)=\max\{\{|\sum_{k=0}^{q} s_{i,k}s_{i,k+t}|:1\le i\le q-2, 1\le t\le q\} \cup \{|\sum_{k=0}^{q} s_{i,k}s_{j,k+t}|: 1\le i\neq j\le q-2, 0\le t\le q\}\}.$
\end{itemize}

\begin{table}[]\label{tab:2}
	\setlength{\abovecaptionskip}{0pt}%
	\setlength{\belowcaptionskip}{10pt}%
	\caption{PARAMETERS OF OUR SEQUENCES}
	\center
	\begin{tabular}{@{}|c|c|c|c|@{}}
		\toprule
		Field Size & Sequence Length & Family Size & Correlation \\ \midrule
		$3^3$    & 28             & 25       &    12                          \\ \midrule
		$3^4$    & 82              & 79        &       22                        \\ \midrule
		$3^5$   & 244            & 241        & 32                            \\ \midrule
		$5^2$ & 26             & 23        & 14                               \\ \midrule
		$5^3$   &126             & 123         &  26                             \\ \midrule
		23 & 24          & 21      &           12               \\ \bottomrule
	\end{tabular}
\end{table}

In order to better understand the above explicit construction, we provide an example for $q=5^2$.
\begin{ex}
{\rm For $q=5^2$, let $\zeta$ be a primitive element of $\F_{25}$ satisfying $\zeta^2+\zeta+2=0$. We choose a primitive quadratic irreducible polynomial $p(x)=x^2+(\zeta+2)x+\zeta+2$, i.e., $a=\zeta+2$ and $b=\zeta+2$.
From Theorem \ref{thm:3.6}, an automorphism $\sigma$ of $E=\F_q(u)$ with order $q+1$ can be determined explicitly by $\sigma(u)=-b/(u+a)$. Let $P_0$ be the zero of $u$ in $\F_q(u)$.
With the help of software Sage, all rational places $P_{\Gs^j}=\sigma^j(P_0)$ with $0\le j\le 25$ of $E$ and all representative elements $z_i$ with $1\le i\le 23$ in equivalence classes of $V\setminus \{0\}$ can be determined explicitly from Proposition \ref{prop:3.10} and Proposition \ref{prop:3.12}.
Hence, the family of binary sequences with a low correlation constructed in Theorem \ref{thm:3.17} can be obtained explicitly as follows:\\
{\footnotesize
$\bs_1=(-1, 1, -1, -1, 1, 1, -1, 1, 1, 1, 1, 1, 1, -1, 1, 1, -1, -1, 1, -1, -1, -1, 1, 1, -1, -1)$, \\
$\bs_2=(-1, 1, -1, 1, -1, -1, 1, 1, 1, 1, 1, -1, -1, -1, 1, -1, -1, -1, 1, 1, 1, 1, 1, -1, -1, 1)$, \\
$\bs_3=(1, 1, 1, 1, 1, -1, -1, -1, -1, 1, -1, -1, -1, -1, 1, -1, -1, -1, -1, 1, -1, -1, -1, -1, 1, 1)$, \\
$\bs_4=(-1, 1, -1, 1, 1, -1, 1, 1, -1, -1, 1, 1, 1, 1, 1, 1, 1, 1, 1, -1, -1, 1, 1, -1, 1, 1)$, \\
$\bs_5=(1, 1, 1, -1, 1, 1, -1, 1, -1, -1, 1, -1, -1, 1, 1, 1, -1, -1, 1, -1, -1, 1, -1, 1, 1, -1)$, \\
$\bs_6=(-1, 1, -1, -1, 1, 1, -1, 1, -1, -1, -1, -1, 1, -1, 1, -1, 1, -1, -1, -1, -1, 1, -1, 1, 1, -1)$, \\
$\bs_7=(-1, 1, 1, -1, -1, -1, -1, 1, 1, -1, -1, 1, -1, 1, -1, 1, 1, 1, 1, 1, 1, -1, 1, -1, 1, -1)$, \\
$\bs_8=(1, 1, -1, 1, -1, -1, 1, -1, -1, 1, -1, -1, 1, -1, 1, 1, 1, -1, -1, -1, 1, 1, -1, -1, -1, 1)$, \\
$\bs_9=(-1, 1, -1, 1, 1, -1, -1, -1, 1, 1, 1, 1, 1, 1, -1, -1, -1, 1, 1, -1, 1, -1, 1, 1, 1, 1)$, \\
$\bs_{10}=(1, 1, 1, -1, -1, 1, 1, 1, 1, -1, 1, -1, -1, -1, -1, 1, -1, 1, 1, 1, 1, -1, -1, 1, 1, 1)$, \\
$\bs_{11}=(-1, 1, -1, -1, -1, 1, -1, -1, 1, -1, -1, -1, 1, -1, -1, 1, 1, 1, 1, -1, -1, 1, 1, 1, 1, -1)$, \\
$\bs_{12}=(-1, 1, -1, -1, -1, 1, 1, -1, 1, 1, -1, -1, -1, 1, -1, 1, 1, -1, 1, 1, 1, 1, 1, -1, 1, 1)$, \\
$\bs_{13}=(-1, 1, 1, 1, 1, 1, -1, 1, 1, -1, -1, 1, 1, 1, 1, -1, -1, -1, 1, 1, 1, 1, -1, -1, 1, 1)$, \\
$\bs_{14}=(1, 1, -1, 1, 1, 1, 1, 1, -1, -1, -1, -1, 1, -1, -1, 1, -1, -1, 1, -1, -1, -1, -1, 1, 1, 1)$, \\
$\bs_{15}=(1, 1, 1, 1, -1, -1, 1, 1, 1, -1, -1, 1, -1, 1, 1, 1, 1, -1, 1, -1, -1, 1, 1, 1, -1, -1)$, \\
$\bs_{16}=(-1, 1, -1, 1, -1, 1, -1, 1, -1, -1, 1, 1, 1, -1, 1, 1, 1, 1, 1, 1, 1, -1, 1, 1, 1, -1)$, \\
$\bs_{17}=(-1, 1, 1, 1, -1, -1, 1, 1, 1, -1, 1, -1, -1, -1, -1, -1, 1, 1, 1, 1, -1, -1, -1, -1, -1, 1)$, \\
$\bs_{18}=(1, 1, -1, -1, -1, -1, -1, -1, -1, -1, -1, -1, 1, 1, 1, -1, 1, -1, -1, 1, 1, -1, -1, 1, -1, 1)$, \\
$\bs_{19}=(1, 1, -1, 1, -1, 1, 1, -1, 1, -1, 1, 1, -1, 1, -1, -1, -1, 1, -1, -1, 1, -1, -1, -1, 1, -1)$, \\
$\bs_{20}=(1, 1, -1, -1, 1, 1, -1, 1, 1, -1, -1, 1, 1, -1, -1, -1, 1, 1, -1, 1, -1, 1, 1, -1, -1, -1)$, \\
$\bs_{21}=(-1, 1, -1, -1, 1, -1, -1, 1, -1, 1, -1, -1, -1, 1, 1, -1, -1, 1, 1, -1, -1, -1, 1, -1, 1, -1)$, \\
$\bs_{22}=(1, 1, 1, -1, -1, 1, 1, 1, 1, 1, 1, 1, -1, -1, 1, -1, -1, -1, -1, 1, -1, -1, 1, 1, 1, 1)$, \\
$\bs_{23}=(1, 1, 1, -1, -1, -1, -1, -1, 1, 1, 1, -1, 1, 1, -1, 1, -1, 1, 1, 1, -1, 1, -1, 1, 1, -1)$.} \\
It is easy to verify that the correlation of this family of sequences is $14=4+2\sqrt{25}$.}
\end{ex}

We list more numerical results on binary sequences with a low correlation obtained from the above explicit construction in the above Table II.

\end{document}